\documentclass[a4paper]{article}
\usepackage[margin=2.5cm]{geometry}

\usepackage{fontenc}[T1]
\usepackage{alltt}
\usepackage{graphicx}

\usepackage[english]{babel}
\usepackage{multirow}
\graphicspath{{figures/}}
\usepackage{hyperref}
\usepackage{amsthm}

\usepackage{multirow}
\usepackage{booktabs}
\usepackage{amsmath}
\usepackage{amsfonts}
\usepackage{amssymb}
\usepackage{hyperref}
\usepackage{mdwlist}
\usepackage[ruled,vlined,linesnumbered]{algorithm2e}
\usepackage{color}
\usepackage{pgf}
\usepackage{longtable}
\usepackage{pdflscape}
\usepackage[round]{natbib} 
\usepackage{mathrsfs}
\usepackage{subcaption}
\usepackage{graphicx}
\usepackage{algorithmic}
\usepackage{rotating}

\newtheorem{theorem}{Theorem}
\newtheorem{lemma}[theorem]{Lemma}
\newtheorem{definition}[theorem]{Definition}
\newtheorem{problem}[theorem]{Problem}

\begin{document}

\title{Column Generation Algorithms for Nonparametric Analysis of Random Utility Models}
\author{\sc Bart Smeulders}
\date{\today}

\maketitle
\begin{abstract}
  \cite{kitamura2012} develop a nonparametric test for linear inequality constraints, when these are are represented as vertices of a polyhedron instead of its faces. They implement this test for an application to nonparametric tests of Random Utility Models. As they note in their paper, testing such models is computationally challenging. In this paper, we develop and implement more efficient algorithms, based on column generation, to carry out the test. These improved algorithms allow us to tackle larger datasets.
\end{abstract}

\section{Introduction} \label{Introduction}
In a recent paper \emph{"Nonparameteric Analysis of Random Utility Models"}, \cite{kitamura2012} (henceforth KS) develop a test for nonparametric testing of Random Utility Models (RUM). They test the hypothesis that a repeated cross-section of demand data might have been generated by a population of rational consumers. A practical implementation of this test leads to a challenging computational problem. The linear program proposed by \cite{mcfadden1990} is extended with a quadratic objective function, minimizing the Euclidean distance. In effect, the minimum distance between a vector and a cone in a high-dimensional space is calculated. This quadratic program must be solved to compute the test statistic and for each bootstrap replication for the simulation of the critical value. This quadratic program is large, with one variable for each rational choice type. The number of such types rises exponentially with the number of choice situations, and even identifying all types is time consuming. In fact, this is the main limiting factor in KS's implementation. \\

The computational problems handled in this paper are similar to those encountered in the study of random utility models in binary choice settings where rational choice types are represented by strict linear orders over the choice alternatives \citep{block1960}. Specifically, \cite{cavagnaro2014} calculates Bayes factors, a measure for model comparison, for the random utility model. Calculating these factors requires numerous checks to test whether a vector lies inside a polytope. In effect, they test whether the Euclidean distance is equal to zero. For small datasets, inequalities are known describing the polytope, but these descriptions grow quickly with the number of choice alternatives and no full description exists for eight or more choice alternatives \cite{Marti2011ch}. \cite{smeulders2018COR} propose algorithms capable of handling larger datasets, by making use of an adaptation of the linear program of \cite{mcfadden1990}. As in the current paper, the large number of rational choice types, and thus variables in the linear program, makes solving the complete model inefficient. By transforming the problem into an optimization problem, through which the point in the polytope minimizing the Manhattan-distance to the vector is found, a column generation approach can be applied. Informally, a column generation approach makes use of the fact that optimal solutions to optimization problems with large numbers of variables, but relatively few constraints, have optimal solutions that only use a relatively small number of variables. A column generation approach starts with a limited number of variables, and identifies new ones as needed through a separate optimization problem, and thus circumvents the problem of having to identify all rational choice types.\\

In this paper, we will use some of the same ideas. We propose a column generation approach for the Euclidean distance calculation.  Furthermore, we note that the tightening procedure in KS is incompatible with column generation, as it requires knowledge of all rational choice types. To overcome this obstacle, we show that a slight modification to the procedure is possible to remove this requirement. To show the practical benefit of the column generation algorithm, we re-analyze the empirical application of \cite{deb2017revealed}(henceforth DKSQ). To increase computation speed, we develop heuristic algorithms to generate interesting rational choice types for this setting. We show the computational improvements make it possible to study much longer budget sequences.\\

The paper unfolds as follows. In section \ref{Sec:RUM}, we briefly describe the Random Utility Model (RUM). Section \ref{Sec:KitSto} lays out the test described in KS, with a focus on the computational problem of calculating the test statistic. Next, in section \ref{Sec:Column Generation} we describe a column generation algorithm, which we use to more efficiently compute the test statistic. Section \ref{Sec:Tight} describes how to handle the tightening of the cone in a manner that is consistent with column generation. An empirical application is contained in Section \ref{Sec:Application}. We first describe the particular model tested in DKSQ. Next, we show how to implement the general column algorithm for this setting.  Finally, we show the computational benefits of our approach.

\subsection{Random Utility Models} \label{Sec:RUM}

We briefly describe the RUM in a discrete choice setting. KS handle a continuous choice setting, as does the application by DKSQ, but both rely on discretization to make testing possible. Consider the set $\mathcal{X}$ of all discrete choice options, we denote individual choice options by $x_i$. Let $u: \mathcal{X} \rightarrow \mathbb{R}$ denote a utility function. For simplicity, we assume $u(x_i) \ne u(x_j)$ for all $i, j \in \mathcal{X}, i \ne j$. A choice situation $t$ is characterized by a subset of the discrete choice options, denoted $\mathcal{X}_t \subseteq \mathcal{X}$.  A rational actor with a utility function $u$ then picks choice option $x$ satisfying
\begin{align*}
  x = \arg \max_{x_j \in \mathcal{X}_t} u(x_j).
\end{align*}
We furthermore denote the choice option $x$ chosen in situation $t$ by $x(t)$.\\

Given the discrete nature of the choice options, there is a finite number of ways an actor can choose over all situations. We characterize a choice type, indexed by $r$, by the choices she makes in each choice situation. Specifically, we encode a choice type $r$ as $\mathbf{a}_r = (a_{r,1,1}, \ldots, a_{r,T,|\mathcal{X}|})$, with $a_{r,t,i} = 1$ if choice option $x_{i}$ is chosen in situation $t$ by type $r$ and $a_{r,t,i} = 0$ otherwise. The set of rational choice types $\mathcal{R}$ is the set of all types $r$ for which there exists some utility function $u_r$ such that
\begin{align*}
  a_{r,t,i} = 1 \text{ if and only if } x_i = \arg \max_{x_j \in \mathcal{X}_t} u_r(x_j)
\end{align*}

Let $P_\mathcal{R}$ be a probability distribution over all rational choice types, and let $p_r$ be the probability of a given choice type. We define the sets $\mathcal{R}_{t,i}$ as the subsets of $\mathcal{R}$ such that $r \in \mathcal{R}_{t,i}$ if and only if $a_{r,t,i} = 1$, i.e. $\mathcal{R}_{t,i}$ is the set of rational choice types which choose $x_i$ in choice situation $t$. Now suppose we observe choices for the given choice situations, with $\pi_{t,i}$ the rate at which option $i$ is chosen in situation $t$.

\begin{definition} \label{def:Stoch Rat}
  The observed choices $\pi$ are stochastically rationalizable if and only if there exists a distribution $P_\mathcal{R}$ over choice types, such that
  \begin{align}
    \sum_{r \in \mathcal{R}_{t,i}} p_r = \pi_{t,i} & & \forall t = 1,\ldots,T,  x_i \in \mathcal{X}. \label{RUM}
  \end{align}
\end{definition}

Before continuing, we would like to highlight the geometric interpretation of Definition \ref{def:Stoch Rat}. Consider a space, with the number of dimensions equal to the sum of the number of choice options available in each choice situation, over all choice situations. We can interpret $\pi$ as a vector in this space, with $\pi_{t,i}$ the coordinate in the dimension associated with $t$ and $i$. Likewise, the vectors $\mathbf{a}_r$ provide coordinates in each dimension for each rational choice pattern. These vectors describe a convex cone, which we denote by $ \mathcal{C}$
\begin{align} \label{V-Representation}
\mathcal{C} = \{\mathbf{c}| \mathbf{c} = \sum_{r \in \mathcal{R}} \lambda_r \mathbf{a}_r, \lambda_r \geq 0 ,~\forall r \in \mathcal{R}\}.
\end{align}
This representation of the cone is called the $V$-representation, as it is based on the vectors defining the cone. Choice probabilities are rationalizable if and only if $\pi \in \mathcal{C}$. \\

Equivalently, there exists a $H$-representation of the cone, based on hyperplanes. Consider the set of hyperplanes $\mathcal{H} = \mathcal{H}^\leq \cup \mathcal{H}^=$. Each $h \in \mathcal{H}^\leq$ divides the space into half-spaces, one of which is the feasible region (which includes the hyperplane), the other infeasible. For each $h \in \mathcal{H}^=$, only the hyperplane itself is the feasible region. The union of these feasible regions is the cone $\mathcal{C}$. Specifically, consider a set of hyperplanes $\mathcal{H}$. For each $h \in \mathcal{H}$, there exist parameters $b_{h,t,i}$ with $\sum_{t}^{T} \sum_{i}^{I_t} b_{h,t,i} c_{t,i} = 0$ describing the hyperplane. Then
\begin{align}\label{H-Represenation}
  \mathcal{C} = \left\{\mathbf{c} \left\vert
  \begin{array}{l}
  \sum_{t}^{T} \sum_{i}^{I_t} \mathbf{b}_{h}~ \mathbf{c} \leq 0, \forall h \in \mathcal{H}^\leq \\
  \sum_{t}^{T} \sum_{i}^{I_t} \mathbf{b}_{h}~ \mathbf{c} = 0, \forall h \in \mathcal{H}^=
  \end{array}
  \right.
  \right\}.
\end{align}

\subsection{Testing the Random Utility Model} \label{Sec:KitSto}
In this section, we briefly lay out the test described by KS, focussing on the computational problems that arise when implementing the test. We refer to KS for a more thorough explanation of the test.\\

\subsubsection{Test Statistic}
Let $\hat{\pi}$ be an estimator for $\pi$. KS propose to use the Euclidean distance between the vector $\hat{\pi}$ and the cone $C$ as the test statistic $J_N$. Formally, $J_N$ is the optimum objective value to the problem (\ref{QP Original Start})-(\ref{QP Original End}). In this problem, $p_r$ denotes the probability associated with type $r$. $s_{t,i}$ denotes the distance, in the dimension associated with patch $i$ on budget $t$, between the linear combination of the types and the estimated choice probabilities $\hat{\pi}_{t,i}$. $N$ is the number of observations over all time periods. Note that if, and only if, $J_N = 0$, $\hat{\pi}$ is stochastically rationalizable in the sense of Definition \ref{def:Stoch Rat}.

\begin{align}\label{QP Original Start}
  \text{Minimize} &  & J_N = N \sum_{t = 1}^{T} \sum_{i = 1}^{I_t} s^2_{t,i} \\
  \text{Subject to} & & \nonumber \\
  & & \sum_{r \in\mathcal{R}_{t,i}} p_r + s_{t,i} & = \hat{\pi}_{t,i} & & \forall x_{t,i} \in \mathcal{X} \\
  & & p_r & \geq 0 & & \forall r \in\mathcal{R} \label{QP Original End}
\end{align}

The projection of $\hat{\pi}$ onto $C$ is denoted by $\hat{\eta} = \sum_{r \in \mathcal{R}} p_r \mathbf{a}_r$.

\subsubsection{Critical Value} \label{Sec:CritVal}
The critical value is computed through a bootstrap procedure, which relies on a tuning paramater $\tau_N$. Given $R$ bootstrap replications with sample frequencies $\hat{\pi}^{*(r)}$ for $r = 1, \ldots, R$, the critical value for $J_N$ is computed as follows.
\begin{enumerate}
\item Obtain the $\tau_N$-tightened estimator $\hat{\eta}^{\tau_N}$, with $\hat{\eta}^{\tau_N} = \sum_{r \in\mathcal{R}} p_r \mathbf{a}_r$, solving (\ref{QP Tight Start})-(\ref{QP Tight End}).
\begin{align}\label{QP Tight Start}
  \text{Minimize} &  & J_N = N \sum_{t = 1}^{T} \sum_{i = 1}^{I_t} s^2_{t,i} \\
  \text{Subject to} & & \nonumber \\
  & & \sum_{r \in\mathcal{R}_{t,i}} p_r + s_{t,i} & = \hat{\pi}_{t,i} & & \forall x_{t,i} \in \mathcal{X} \\
  & & p_r & \geq \tau_N / |\mathcal{R}| & & \forall r \in\mathcal{R} \label{QP Tight End}
\end{align}
\item Define the $\tau_N$-tightened recentered bootstrap estimators.
\begin{align}
\hat{\pi}_{\tau_N}^{*(r)} = \hat{\pi}^{*(r)} - \hat{\pi} + \hat{\eta}_{\tau_N}.
\end{align}
\item The bootstrap test statistics $J_N^{*(r)}(\tau_N)$ are the solutions to (\ref{QP Tight Start})-(\ref{QP Tight End}), using $\hat{\pi}^{*(r)}$ for the right-hand sides of the inequalities.
\item Use the empirical distribution of $J_N^{*(r)}(\tau_N)$, $m = 1, \ldots, M$ to obtain the critical value for $J_N$.
\end{enumerate}

\subsection{Computational Difficulties}
To compute the tests statistic $J_N$ and to obtain a critical value for it, the problem (\ref{QP Original Start})-(\ref{QP Original End}) must be solved once, and the problem (\ref{QP Tight Start})-(\ref{QP Tight End}) solved $1 + M$ times (once to obtain the $\tau_N$-tightened estimator, and then once for each bootstrap replication). As mentioned by KS, solving these problems is computationally challenging. The straightforward approach implemented by KS requires that each rational choice type $r \in \mathcal{R}$ is first identified (though this must be done only once), and then a large quadratic program must be solved. The number of rational choice types can however rise exponentially with the number of periods considered. This makes the approach by KS computationally costly for moderately sized instances, and makes larger instances impossible. Table \ref{table:rationalchoicetypes} shows the approximate number of rational choice types for different size instances in the application of DKSQ.\footnote{The number of total choice types is calculated exactly, random sampling is used to estimate the ratio of rational choice types to total choice types.}

\begin{table}[htbp]
  \centering
    \begin{tabular}{r|ll|ll|ll|}
          & \multicolumn{2}{c|}{3 Goods} & \multicolumn{2}{c|}{4 Goods} & \multicolumn{2}{c|}{5 Goods} \\
          & \multicolumn{1}{c}{Min} & \multicolumn{1}{c|}{Max} & \multicolumn{1}{c}{Min} & \multicolumn{1}{c|}{Max} & \multicolumn{1}{c}{Min} & \multicolumn{1}{c|}{Max} \\ \hline
    6 Periods  & $3.00*10^{1}$ & $5.44*10^{4}$ & $1.38*10^{3}$ & $4.30*10^{5}$ & $1.38*10^{3}$ & $4.30*10^{5}$ \\
    10 Periods & $5.14*10^{3}$ & $3.35*10^{8}$ & $1.52*10^{8}$ & $1.03*10^{13}$ & $6.76*10^{12}$ & $6.93*10^{16}$ \\
    15 Periods & $2.12*10^{9}$ & $3.87*10^{15}$ & $6.76*10^{17}$ & $7.05*10^{21}$ & $1.05*10^{19}$ & $4.07*10^{22}$ \\
    20 Periods & $2.01*10^{18}$ & $2.98*10^{22}$ &       &       &       &  \\
    \end{tabular}
  \caption{Approximate maximum and minimum number of rational choice types in the DKSQ application.}
  \label{table:rationalchoicetypes}
\end{table}

In the following sections, we describe how these problems can be solved without requiring the identification of all choice types, by making use of a column generation algorithm. In section \ref{Sec:Column Generation}, we handle the problem (\ref{QP Original Start})-(\ref{QP Original End}). Problem (\ref{QP Tight Start})-(\ref{QP Tight End}) is subtly different, requiring a strictly positive lower bound on the variables $p_r$, associated with the choice types. Solving this problem without identifying all choice types is thus not possible. However, in section \ref{Sec:Tight} we propose minor changes to the KS-procedure for obtaining the critical value, so that these strictly positive lower bounds are no longer necessary.

\section{Euclidean Projection through Column Generation} \label{Sec:Column Generation}

We will tackle this problem by making use of a column generation algorithm. Instead of solving (\ref{QP Original Start})-(\ref{QP Original End}) directly, we will start with a limited version of this problem, using only a small set of its variables. We will call this problem the {\it restricted master}. Given a solution to this problem, we find a hyperplane, separating the vector $\pi$ from the restricted polytope. In a second problem, called the {\it pricing problem} we check whether there exists any point of the full polytope on the side of $\pi$ of the separating hyperplane. If no such point exists, we show the solution to the restricted master is also a solution to (\ref{QP Original Start})-(\ref{QP Original End}). If such a point does exist, we add the corresponding variable to the restricted master and (re-)solve this problem. Such an approach to computing the distance between a point and a polytope, by iteratively taking into account additional vertices of a polytope, is originally described by \cite{wolfe1976}. Wolfe does make use of an exhaustive list of vertices of the polytope, which is impractical given the large number of vertices in our application. \cite{cadoux2010} extends this to a setting without an exhaustive list of vertices. \\

Let us look at the proposed algorithm step-by-step. First, we solve problem (\ref{QP Original Start})-(\ref{QP Original End}) with a restricted set $\bar{\mathcal{R}}$ of $k$ choice patterns. The {\it bar} notation signifies that the variables, sets or solution belongs to a restricted master problem. From the optimal solution to this restricted problem, $\bar{\mathbf{p}}^* = (\bar{p}^*_1, \ldots, \bar{p}^*_k)$ and $\bar{\mathbf{s}}^* = (\bar{s}^*_{1,1}, \ldots, \bar{s}^*_{T,I_T})$, we can construct the Euclidean projection of $\pi$ on the restricted cone $\bar {\mathcal{C}}$. This projection is the vector $\bar{\mathbf{v}}^* = (\bar{v}^*_{1,1}, \ldots, \bar{v}^*_{T,I_T})$ with $\bar{v}^*_{t,i} = \sum_{r \in \bar{\mathcal{R}}_{t,i}} \bar{p}^*_r$. Now consider the characterization of a Euclidean projection on a convex set (in this case the cone $\mathcal{C}$)

\begin{theorem}
$\mathbf{v}^*$ is the Euclidean projection of $\hat{\pi}$ on $\mathcal{C}$ if and only if  $(\hat{\pi} - \mathbf{v}^*) \cdot (\mathbf{v} - \mathbf{v}^*) \leq 0 $, for all $\mathbf{v} \in \mathcal{C}$.
\end{theorem}

Since $\mathcal{C}$ is the set of linear combinations of vectors $\mathbf{a}_r, r \in \mathcal{R}$, we can also state the following result:

\begin{theorem}
$\mathbf{v}^*$ is the Euclidean projection of $\pi$ on $\mathcal{C}$ if and only if  $(\hat{\pi}- \mathbf{v}^*) \cdot (\mathbf{a}_r - \mathbf{v}^*) \leq 0 $, for all $r \in \mathcal{R}$.
\end{theorem}

\begin{proof}
  Suppose that $(\hat{\pi}- \mathbf{v}^*) \cdot (\mathbf{a}_r - \mathbf{v}^*) \leq 0 $, for all $r \in \mathcal{R}$, we now argue that $(\hat{\pi}- \mathbf{v}^*) \cdot (\mathbf{v} - \mathbf{v}^*) \leq 0 $, for all $\mathbf{v} \in \mathcal{C}$. For each $\mathbf{v} \in \mathcal{C}$, there exist non-negative numbers $\lambda_r$ such that $\mathbf{v} = \sum_{r \in \mathcal{R}} \lambda_r \mathbf{a}_r$. Thus, $(\hat{\pi}- \mathbf{v}^*) \cdot (\mathbf{v} - \mathbf{v}^*)$ can be written as $(\hat{\pi}- \mathbf{v}^*) \cdot (\sum_{r \in \mathcal{R}} \lambda_r \mathbf{a}_r - \mathbf{v}^*)$ or as $\sum_{r \in \mathcal{R}} \lambda_c (\hat{\pi}- \mathbf{v}^*) \cdot ( \mathbf{a}_r - \mathbf{v}^*)$. Since $(\hat{\pi}- \mathbf{v}^*) \cdot (\mathbf{a}_r - \mathbf{v}^*) \leq 0 $, for all $r \in \mathcal{R}$, we also have $\sum_{r \in \mathcal{R}} \lambda_c (\hat{\pi}- \mathbf{v}^*) \cdot ( \mathbf{a}_r - \mathbf{v}^*) < 0$.
\end{proof}

Note that $(\hat{\pi}- \mathbf{v}^*) = \mathbf{s}^*$, thus we can rewrite $(\hat{\pi}- \mathbf{v}^*) \cdot (\mathbf{a}_r - \mathbf{v}^*) \leq 0 $ as $\mathbf{s}^* \mathbf{a}_r \leq \mathbf{s}^* \mathbf{v}^*$. Given this result, we can check whether $\bar{\mathbf{v}}^*$ is the Euclidean projection of $\pi$ on $\mathcal{C}$ by solving the following problem:
\begin{problem}
Does there exist a choice pattern $r \in \mathcal{R}$, such that $\bar{\mathbf{s}}^* \mathbf{a}_r \geq \bar{\mathbf{s}}^*\bar{\mathbf{v}}^*$ ?
\end{problem}

To answer this question, we solve a different optimization problem, usually referred to as the {\it pricing} problem.
\begin{align}
\arg \max_{r \in \mathcal{R}} \bar{s}^* \mathbf{a}_r. \label{pricing problem}
\end{align}

It is clear that if we find an optimal solution to (\ref{pricing problem}), we can easily check whether it satisfies the threshold value $(\bar{\mathbf{s}}^* \bar{\mathbf{v}}^*)$. If the threshold is met, the choice type is added to the set of choice patterns considered in the restricted problem, which is then re-solved. Otherwise, the solution $\bar{\mathbf{p}}^*, \bar{\mathbf{s}}^*$ to the restricted problem is also the optimal solution to the problem considering the full set of choice patterns. \\

Although an optimal solution to (\ref{pricing problem}) is preferable, it is important to note that any $r \in \mathcal{R}$ with $\bar{\mathbf{s}} \mathbf{a}_r \geq \bar{\mathbf{s}}^* \bar{\mathbf{v}}^*$ is sufficient to continue with the column generation. To speed up computation, it can thus be more interesting to quickly find any type $r \in \mathcal{R}$ meeting the threshold than to spend a longer time finding the solution to (\ref{pricing problem}). Algorithm \ref{QP Algorithm} summarizes the column generation algorithm.

\begin{algorithm} \caption{Quadratic Program Column Generation Algorithm}
\label{QP Algorithm}
\begin{algorithmic}[1]
\STATE Solve Initial Restricted Master Problem, optimal solution $\bar{p}^*, \bar{s}^*, \bar{v}^*$. \label{master}
\WHILE {there exists $r \in \mathcal{R}$ with $\bar{s}^* a_r \geq \bar{s}^*\bar{v}^*$}
\STATE Find a choice pattern $r \in \mathcal{R}$ with $\bar{s}^* a_r \geq \bar{s}^*\bar{v}^*$.
\STATE Set $\bar{\mathcal{R}} := \bar{\mathcal{R}} \cup d$.
\STATE Re-Solve Restricted Master Problem, optimal solution $\bar{p}^*, \bar{s}^*, \bar{v}^*$.
\ENDWHILE
\STATE Restricted Master Solution $\bar{p}^*, \bar{s}^*, \bar{v}^*$ is the optimal solution $y^*, s^*, v^*$ to the Complete Master Problem.
\end{algorithmic}
\end{algorithm}

This approach allows us to solve (\ref{QP Original Start})-(\ref{QP Original End}) with only a fraction of the rational choice types identified, as we will show in the application.
\section{Solving Tightened Problems} \label{Sec:Tight}
In problems of the form (\ref{QP Tight Start})-(\ref{QP Tight End}), there is the additional complication that there is a strictly positive lower bound on $p_r$ for all $r \in \mathcal{R}$. This is incompatible with the column generation algorithm described in the previous section, which only uses a subset of these variables. We work around this problem in two steps. First, we show that for every problem of the form (\ref{QP Tight Start})-(\ref{QP Tight End}), with strictly positive lower bounds, there exists an equivalent problem with zero lower bounds which can be solved using the column generation algorithm. If all rational choice types have a strictly positie lower bound, finding this equivalent problem still requires knowledge of all rational choice types. However, we also show that the tightening can be achieved by setting strictly positive lower bounds for only a subset of the rational choice types.

\begin{lemma}
The problem
\begin{align}\label{QP Tight Eq Start}
  \text{Minimize} &  & J_N = N \sum_{t = 1}^{T} \sum_{x_i \in \mathcal{X}} s^2_{t,i}  \\
  \text{Subject to} & & \nonumber \\
  & & \sum_{r \in\mathcal{R}_{t,i}} p_r + s_{t,i} & = \hat{\pi}_{t,i} - \sum_{r \in\mathcal{R}_{t,i}} \tau_N / |\mathcal{R}| & & \forall x_{t,i} \in \mathcal{X} \\
  & & p_r & \geq 0&  & \forall r \in\mathcal{R} \label{QP Tight Eq End}
\end{align}
is equivalent to problem (\ref{QP Tight Start})-(\ref{QP Tight End}).
\end{lemma}
\begin{proof}
 Given a feasible solution $(s_{t,i}, p_r)$ to (\ref{QP Tight Start})-(\ref{QP Tight End}), $(s_{t,i}, p'_r = p_r - \tau_N / |\mathcal{R}|)$ is a feasible solution to (\ref{QP Tight Eq Start})-(\ref{QP Tight Eq End}). Since both problems have the same objective function, and the $s_{t,i}$ variables have the same value in both feasible solutions, a solution to (\ref{QP Tight Start})-(\ref{QP Tight End}) implies the existence of a solution to (\ref{QP Tight Eq Start})-(\ref{QP Tight Eq End}) with the same objective value. Likewise, given a feasible solution $(s'_{t,i}, p'_r)$ to (\ref{QP Tight Eq Start})-(\ref{QP Tight Eq End}), $(s'_{t,i}, p_r = p'_r + \tau_N / |\mathcal{R}|)$ is a feasible solution to (\ref{QP Tight Start})-(\ref{QP Tight End}), again with the same objective value. Thus, the optimal solution to both problems will have the same value.
\end{proof}

Stating the equivalent problem still requires knowledge of all rational choice types to adjust the right hand side of the constraints. We therefore propose a tightening based on only a subset of the rational choice types.\\

Consider a subset of the rational choice types $\mathcal{R}' \subset \mathcal{R}$, such that for each hyperplane $h \in \mathcal{H}^\leq$, there exists at least one $r \in \mathcal{R}'$ such that $\sum_{t=1}^{T}\sum_{i=1}^{I_t} b_{h,t,i} a_{r,t,i} < 0$. The proof of Lemma 4.1 in KS can be applied to prove the following Lemma.

\begin{lemma}
Define
 \begin{align*}
\mathcal{C} = \{\sum_{r \in \mathcal{R}} \lambda_r a_r| \lambda_r \geq 0 ,~\forall r \in \mathcal{R}\}.
\end{align*}
and let
\begin{align*}
  \mathcal{C} = \left\{c \left\vert
  \begin{array}{l}
  \sum_{t}^{T} \sum_{i}^{I_t} b_{h,t,i}~ c_{t,i} \leq 0, \forall h \in \mathcal{H}^\leq \\
  \sum_{t}^{T} \sum_{i}^{I_t} b_{h,t,i}~ c_{t,i} = 0, \forall h \in \mathcal{H}^=
  \end{array}
  \right.
  \right\}.
\end{align*}
be its $H$-representation. For $\tau > 0$, define
\begin{align*}
\mathcal{C} = \left\{\sum_{r \in \mathcal{R}} \lambda_r a_r \left\vert
 \begin{array}{ll}
 \lambda_r \geq 0 ,&\forall r \in \mathcal{R}\backslash \mathcal{R}' \\
 \lambda_r \geq \tau/|\mathcal{R}'|, &\forall r \in \mathcal{R}'
 \end{array}
 \right.
 \right\}.
\end{align*}
Then one also has
\begin{align*}
  \mathcal{C} = \left\{c \left\vert
  \begin{array}{ll}
  \sum_{t}^{T} \sum_{i}^{I_t} b_{h,t,i}~ c_{t,i} \leq -\tau\phi_h,&\forall h \in \mathcal{H}^\leq \\
  \sum_{t}^{T} \sum_{i}^{I_t} b_{h,t,i}~ c_{t,i} = 0,&\forall h \in \mathcal{H}^=
  \end{array}
  \right.
  \right\}.
\end{align*}
with $\phi_h > 0$ for all $h \in \mathcal{H}^\leq$.
\end{lemma}

This result follows immediately from the proof of Lemma 4.1 in KS.\\

\subsection{Bounds}\label{Sec:Bounds}
Note that identifying the exact distribution of $J_N^{*(m)}(\tau_N)$ is unnecessary, as only the ratio of bootstrap test statistics larger and smaller than $J_N$ is necessary to check whether it falls above or below the critical value. This can be exploited by making use of bounds on the bootstrap test statistics $J_N^{*(m)}(\tau_N)$, to more quickly determine the $p$-value. Specifically, if at any point in the column generation algorithm we can determine, for a given bootstrap repetition, that $J_N^{*(m)}(\tau_N)$ is either
strictly larger or smaller than $J_N$, we terminate the algorithm. In this case we save the lower, or respectively the upper bound. This approach saves time, since bootstrap test statistics must not be computed exactly, but the resulting p-values do not change. \\

Calculating a upper bound on the bootstrap test statistic is straightforward. The objective value of the restricted master problem is immediately a upper bound on the objective value of the complete master problem, as any solution to the restricted master is also feasible for the complete master. Since a restricted master is already solved in every iteration of the column generation algorithm, no additional work is required to obtain this upper bound.\\

Lower bounds on the bootstrap test statistic can be obtained based on (optimal) solutions to the pricing problem as follows. Consider a pricing problem (\ref{con: Price Assign})-(\ref{con: Pricing End}) with objective function $\sum_{t = 1}^{T} \sum_{i = 1}^{I_t} s_{t,i}a_{t,i}$, and let the optimal solution value to this pricing problem be $z^*$. In this case, for each $r \in \mathcal{R}$, we have  $\sum_{t = 1}^{T} \sum_{i = 1}^{I_t} s_{t,i}a_{r,t,i} \leq z^*$. Since the cone $\mathcal{C}$ is the set of linear combinations of the vectors $a_r$, this in turn implies that
\begin{align}
\sum_{t = 1}^{T} \sum_{i = 1}^{I_t} s_{t,i}c_{t,i} \leq z^*, \forall c \in \mathcal{C}.
\end{align}
By solving the following, relatively simple, quadratic optimization problem we thus obtain a lower bound on the bootstrap test statistic.
\begin{align}\label{QP 2 Start}
  \text{Minimize} &  &\sum_{t = 1}^{T} \sum_{i = 1}^{I_t} v^2_{t,i} \\
  \text{Subject to} & & \nonumber \\
  & & c_{t,i} + v_{t,i} & = \hat{\pi}^{*(r)}_{t,i} & & \forall x_{t,i} \in \mathcal{X} \\
  & & \sum_{t = 1}^{T} \sum_{i = 1}^{I_t} s_{t,i}c_{t,i} & \leq z^* & \label{QP 2 End}
\end{align}
Note that obtaining this lower bound requires an optimal solution to the pricing problem. This leads to the following trade-off, where using heuristics leads to a more efficient column generation algorithm, with less time spent per iteration. On the other hand, solving the pricing problem to optimality allows us to obtain a lower bound on the test statistic, which may end the column generation algorithm outright. In our implementation, we will only compute the lower bound in iterations for which the pricing problem was solved exactly.

\section{Empirical Application} \label{Sec:Application}
For our empirical application, we use our improved algorithms to replicate the tests performed by \cite{deb2017revealed}. We first briefly summarize the model tested by these authors in Section \ref{sec:GAPP}. We focus on those aspects that are important to the computational problem of calculating the test statistic and critical value. To apply the general column generation approach described in previous sections to the setting of the application, some customization is required. This is described in Section \ref{sec:Pricing Problem}.  \\

\subsection{Generalized Axiom of Revealed Price Preference} \label{sec:GAPP}
Consider a dataset $\mathcal{D} = \{(\mathbf{p}_t,\mathbf{q}_t)\}_{t=1}^T$, with $\mathbf{q}_t \in \mathbb{R}^L_+$ a bundle of $L$ goods    bought at price vector $\mathbf{p}^t \in \mathbb{R}^L_{++}$. We are interested in the preference of the consumer over prices. Suppose, that for observations $t$ and $t'$, we have $\mathbf{p}_{t'}\mathbf{q}_t < \mathbf{p}_{t}\mathbf{q}_t$. In this case, the consumer would prefer prices $\mathbf{p}_{t'}$ over prices $\mathbf{p}_t$, since the former allows the consumer to purchase the same bundle of goods, and have (more) money left over to spend in other ways. Formally, we denote $\mathbf{p}_{t'}\mathbf{q}_t <(\leq) \mathbf{p}_{t}\mathbf{q}_t$ by $\mathbf{p}_{t'} \succ_p(\succeq_p) \mathbf{p}_t$. Furthermore, we denote the relation $\mathbf{p}_{t'} \succeq_p^* \mathbf{p}_t$ if there exists a chain of price vector such that $\mathbf{p}_{t'} \succeq_p \ldots \succeq_p \mathbf{p}_t$, and $\mathbf{p}_{t'} \succ_p^* \mathbf{p}_t$ if such a chain exists with at least one $\succ_p$ relation included.
\begin{definition}
The dataset $\mathcal{D} = \{(\mathbf{p}_t, \mathbf{q}_t)\}_{t=1}^T$ satisfies the Generalized Axiom of Revealed Price Preference (GAPP) if there do not exist two observations $t, t' \in T$, such that $\mathbf{p}_{t'} \succeq_p^* \mathbf{p}_t$ and $\mathbf{p}_{t} \succ_p^* \mathbf{p}_{t'}$
\end{definition}

Now consider a augmented utility function $u(\mathbf{q},-\mathbf{pq})$. Note that the utility depends both on the bundle of goods, as well as on the amount of money expended. DKSQ prove the following theorem.
\begin{theorem}
Given a dataset $\mathcal{D} = \{(\mathbf{p}_t,\mathbf{q}_t)\}_{t=1}^T$, the following are equivalent:
\begin{enumerate}
  \item $\mathcal{D}$ can be rationalized by an augmented utility function.
  \item $\mathcal{D}$ satisfies GAPP.
  \item $\mathcal{D}$ can be rationalized by an augmented utility function that is strictly increasing, continuous and concave. Moreover, $u$ is such that $\max_{\mathbf{q} \in \mathbb{R}^L_+} u(\mathbf{q},-\mathbf{pq})$ has a solution for all $\mathbf{p} \in \mathbb{R}^L_{++}$.
\end{enumerate}
\end{theorem}

This model can be discretized, which is necessary to employ the algorithms discussed previously. For each period $t= 1,\ldots,T$, the set of possible choices $(\mathbb{R}^L_+)$ is partitioned into subspaces $x_{t,1}, \ldots, x_{t,I_t}$, we use the word ``patch" to refer to these elements. This partitioning is such that (1) $\mathbb{R}^L_+ = \bigcup_{i=1}^{I_t} x_{t,i}$ and (2) for all bundles $\mathbf{q}, \mathbf{q}' \in x_{t,i}$ and each other period $t'$,  $\mathbf{q}$ and $\mathbf{q}'$ induce the same revealed preference relations, (3) the partition is of minimal size. Analogous to KS and DKSQ, we only consider patches corresponding to strict price preference relations. $\mathcal{X}_t$ denotes the set of patches of periods $t$, while $\mathcal{X}$ is the set of all patches. Note that for each time periods, the number of patches, and thus possible choices we must account for, is now bounded from above by $2^T$.\\

Instead of all possible utility functions, we only consider rational choice types. We encode a choice type $r$ as $\mathbf{a}_r = (a_{r,1,1}, \ldots, a_{r,T,I_T})$, with $a_{r,t,i} = 1$ if the patch $x_{t,i}$ chosen at time $t$ by type $r$ and $a_{r,t,i} = 0$ otherwise. The set of rational choice types $\mathcal{R}$ is the set of all types $r$ for which the chosen patches induce price preference relations satisfying GAPP. We furthermore define the sets $\mathcal{R}_{t,i} := \{r \in\mathcal{R}| a_{r,t,i} = 1\}$. Given that there exists a finite number of patches, the number of rational choice types to be considered is also finite.

\subsection{Setting Specific Pricing Problem}\label{sec:Pricing Problem}
In the previous sections, we lay out how the column generation approach can be used to calculate the test statistic and how to handle the tightening procedure. This description is given in a general way, without any reference to a specific setting. In case of the master problem, this is not necessary. Formulation (\ref{QP Original Start})-(\ref{QP Original End}) can be used for any discrete choice setting. However, the set of rational choice types $\mathcal{R}$ is determined by the setting. The pricing problem (\ref{pricing problem}), must thus also be tailored to it. In this section, we formulate a pricing problem to test GAPP, and discuss ways to solve it efficiently.\\

The binary variable $\alpha_{t,i}$ indicates which patch is chosen on each budget. $\alpha_{t,i} = 1$ if patch $x_{t,i}$ is chosen from $\mathcal{X}_t$, and $\alpha_{t,i} = 0$ otherwise. The binary variables $\rho_{t,j}$ represent the preference relations between $\mathbf{p}_t$ and $\mathbf{p}_j$. If the patch chosen in $\mathcal{X}_t$ induces $\mathbf{p}_t \succ_p\mathbf{p}_j$, then $\rho_{t,j} = 1$, otherwise $\rho_{t,j} = 0$. $X_{t,i,j}$ is a parameter indicating the price preferences induced by the choice of patch $x_{t,i}$, with $X_{t,i,j} = 1$ if $x_{t,i}$ induces $\mathbf{p}_j \succ_p\mathbf{p}_t$ and $X_{t,i,j} = 0$ otherwise.

\begin{align}
  \text{Maximize} &  &\sum_{t = 1}^{T} \sum_{i = 1}^{I_t} s_{t,i}\alpha_{t,i} \label{Pricing Start}\\
  \text{Subject to} & & \nonumber \\
  & & \sum_{i=1}^{I_t} \alpha_{t,i} & = 1 & & \forall t = 1, \ldots, T \label{con: Price Assign}\\
  & & \sum_{i = 1}^{I_t}\alpha_{t,i}X_{t,i,j} - \rho_{j,t} & \leq 0& & \forall j,t = 1, \ldots, T \label{con: Price Direct}\\
  & & \rho_{j,t} + \rho_{t,k} - \rho_{j,k} & \leq 1 & & \forall k,j,t = 1, \ldots, T \label{con: Price Trans}\\
  & & \rho_{j,t} + \rho_{t,j} & \leq 1 & & \forall j,t = 1, \ldots, T \label{con: Price SARP}\\
  & & \rho_{j,t} & \in \{0,1\} & & \forall j,t,= 1, \ldots, T \label{con: Price Bin 1}\\
  & & \alpha_{t,i} & \in \{0,1\} & & \forall t = 1, \ldots, T,  i = 1, \ldots, I_t \label{con: Pricing End}
\end{align}

Constraint (\ref{con: Price Assign}) ensures exactly one patch is chosen on each budget. Constraints (\ref{con: Price Direct})-(\ref{con: Price SARP}) ensure GAPP is satisfied for the chosen patches.  First, Constraint (\ref{con: Price Direct}) ensures that if a chosen patch induces a price preference relation ($X_{t,i,j} = 1$), $\rho_{j,t}$ must also be set to one. Constraint (\ref{con: Price Trans}) makes sure that the $\rho$-variables also reflect the transitivity of the preference relations. Finally, Constraint (\ref{con: Price SARP}) enforces that the preference relation is acyclic. Together, these constraints enforce that the $a_{t,i}$ variables encode a valid choice pattern that is consistent with GAPP. An optimal solution to this integer program shows whether or not a rational choice type exists that can be added to the master problem, if so, the $a_{t,i}$ variables encode one such type.\\

As mentioned earlier, an optimal solution to the pricing problem is not necessary to advance the column generation algorithm. Any rational choice type for which $\bar{s} a_r \geq \bar{s}^* \bar{v}^*$ can be added to the restricted master problem to obtain a better solution. Since solving the pricing problem to optimality is often computationally costly, we propose to solve the pricing problem using heuristics, which are usually much faster. Only if we can not identify new choice types to add to the restricted master using the heuristic algorithms, will we use exact algorithms. Algorithm \ref{Alg: Pricing} shows how the heuristic and exact approaches work together. In the implementation, we use a {\it best insertion algorithm} \citep{LOP} adapted for this particular problem. A detailed description of the implemented heuristic can be found in Appendix B.\\

\begin{algorithm} \caption{Solving the Pricing Problem}
\label{Alg: Pricing}
\begin{algorithmic}[1]
\STATE Solve the pricing problem using heuristical algorithms.
\IF {The best solution has a value $< \bar{s}^*\bar{v}^*$}
\STATE Solve the pricing problem using exact algorithms.
\ENDIF
\end{algorithmic}
\end{algorithm}

The tightening procedure requires a subset of the rational choice types to be identified a priori. This set $\mathcal{R}'$ can be generated by randomly drawing choice types, testing whether they are rational and then keeping only rational choice types so generated. If the probability that a randomly chosen choice type is rational is low, this approach can be time consuming. To speed up the process, we opted for a semi-random method. In this method, we first randomly generate choice types. Small changes are then made to these choice types to remove violations of rationality. In the application, we set the size of the subset to 1,000 rational choice types. A detailed description of the procedure can be found in Appendix C.

\subsection{Results} \label{sec:Results}
The column generation algorithm described in the previous sections is implemented in C++, and CPLEX 12.8 is used to solve both the quadratic master problem, as well as the exact pricing problems. Computational experiments were run on a computer with a quad-core 2.6 GHz processor and 16Gb RAM. For the first bootstrap iteration, we initialize the set $\bar{\mathcal{R}}$ as an empty set. At the end of each bootstrap iteration, the set $\bar{\mathcal{R}}$ is saved and used as the starting set for the next bootstrap iteration. This approach generally speeds up computation, as good solutions for different bootstrap iterations usually have rational choice types in common, which do not need to be re-generated using these starting sets.\footnote{The set $\bar{\mathcal{R}}$ can become large over time, slowing down computation. If this is the case, it can be beneficial to record how often variables are used in the optimal solution and to periodically remove rarely used variables.}\\

In this section, we discuss the speed-ups that are achieved through the use of various techniques discussed above. Specifically, we iteratively compare the following configurations:
\begin{enumerate}
\item All pricing problems solved exactly, no use of bounds.
\item Heuristic \& exact algorithms for the pricing problem, no use of bounds.
\item Heuristic \& exact algorithms for the pricing problem, Upper Bound used.
\item Heuristic \& exact algorithms for the pricing problem, Upper \& Lower Bound used.
\end{enumerate}

These algorithms are applied to the U.K. Family Expenditure Survey. Table \ref{Table:3Goods} contains the minimum, maximum and average computation time for these configurations over the different instances for a given number of periods. Computation times were capped at 1 hour (3600 seconds) for each instance.

\begin{table}[ht!]
\centering
\begin{tabular}{l|r|r|r|r|r|r|r|r|r|r|r|r|}
&\multicolumn{3}{c}{Exact}&\multicolumn{3}{|c}{Heur. - No Bounds}&\multicolumn{3}{|c}{Heur.- UB}&\multicolumn {3}{|c|}{Heur.- All Bounds} \\
&{Min}&{Avg}&{Max}&{Min}&{Avg}&{Max}&{Min}&{Avg}&{Max}&{Min}&{Avg}&{Max}\\ \hline
6 Periods& 1 & 7 & 13 & 1 & 6 & 11 & 2 & 4 & 10 & 1 & 4 & 9 \\
10 Periods& 11 & 103 & 335 & 9 & 42 & 118 & 5 & 29 & 82 & 4 & 27 & 76 \\
15 Periods& 249 & NA & $>$ 3600 & 81 & 1372 & 3557 & 44 & 643 & 1559 & 42 & 565 & 1327 \\
\end{tabular}
\caption{Minimum, Maximum and Average computation times for 3 goods.}
\label{Table:3Goods}
\end{table}

AS DSKQ report that their current techniques do not allow the testing of more than 8 periods, it is clear that even in the simplest configuration, the column generation algorithm allows the testing of much larger datasets than are possible using the approach by KS and DKSQ. Table \ref{Table:3Goods} furthermore show the large impact the use of heuristics for the pricing problem and the use of bounds has on total computation time. While the influence is limited for the smaller instances, the addition of heuristics lowers average computation time by almost 60\% for the 10 periods instances. For 15 periods, this decrease is nearly 75\% for the instances which finished in both configurations. Likewise, the use of bounds to terminate computation earlier speeds up computation considerably, with a decrease in computation time of 35\% for 10 periods and nearly 60\% for 15 periods. Most of this speed-up is due to the lower bound, though for the 15 periods instances the addition of upper bounds lowered computation times by an additional 12\%. While no instances of 20 periods finished within 1 hour, 142 bootstrap iterations were finished for the hardest instance, suggesting all instances could be finished within about 7 hours. For the full 25 periods dataset, 4.25 hours were necessary to complete 100 bootstrap iterations. \\

In the instances we tested, increasing the number of goods generally increased the computational difficulty of the problem. A higher number of goods led to higher numbers of patches, which in turn increased the number of (rational) choice types. Table \ref{table:rationalchoicetypes} clearly shows this. The increased difficulty is also clearly noticeable in the computation times. Whereas the 10 period instances for 3 goods are solved in 27 seconds on average, this was a lower bound for the 4 good, 10 periods instances. For 4 goods, 1 out of 16 instances was not finished within 1 hour, the other 15 took less than 8 minutes on average. 5 Good instances have slightly higher, but comparable, computation times. Due to the higher difficulty of the 4 and 5 good instances, larger instances still take significant amounts of time. For 15 periods, only 6 bootstrap repetitions were complete for the hardest instance, implying about 6 days of total computation time for 1000 bootstraps.

\section{Discussion}
In this paper, we have shown that while the approach to testing random utility models developed by KS and DKSQ is computationally challenging, advanced algorithms allow for tests of far larger datasets. A main ingredient is to avoid complete enumeration of rational choice types, but generate these as necessary. Applying these algorithms to a model of consumption developed by DKSQ and empirical data from the U.K., we show that the model is supported by the data even over longer periods of time. 

\bibliographystyle{plainnat}
\bibliography{BibRP}

\newpage
\section*{Appendix A: Heuristic Pricing Algorithm}
For the pricing problem we use a \emph{Best Insertion} heuristic to quickly generate good rational choice types to add to the restricted master problem. The Best Insertion Algorithm iteratively creates an ordering of the time periods, which (can) correspond to a rational choice type. First, we explain the link between orderings of the timer periods and rational choice types. Next, we explain how to build an ordering that provides a good solution to the pricing problem. Algorithm \ref{Alg: Best Insert} provides the pseudo-code for the algorithm. \\

Consider an ordering $O_T$ over all $T$ time periods $t \in \mathcal{T}$. We can associate a rational choice type with this ordering if the patch chosen in a lower ranked time period is not preferred over one chosen in a higher ranked time period. More specifically, let $o_T(t)$ be the position of time period $t$ in the ordering. We can associate a rational choice type with this ordering if for each time period $t$ there exists a patch $x_{t,i}$, which lies above all budget planes $\mathcal{B}_{t'}$ with $o_T(t) < o_T(t') \leq T$ (i.e. $X_{t,i,t'} = 1)$. Notice that in this case, there exists a feasible solution to the pricing problem for which $\rho_{q_j,q_{j'}} = 1$ only if $j \leq j'$. Given the objective function of the pricing problem, we can easily find the objective value of the best rational choice types respecting the ordering of time periods using the following function.

\begin{align}
V(O_T) = \sum_{t = 1}^{T} \max_{(i: X_{t,i,t'} = 1, \forall t' \text{ for which } o_T(t) < o_T(t'))} s_{t,i}. \label{Eq:Order Eval}
\end{align}

With $s_{t,i}$ the value of choosing patch $x_{t,i}$ in the pricing problem.\\

Building an ordering is done in an iterative fashion. Consider an ordering $O_m$ of $m$ time periods in the set $\mathcal{T'} \subset \mathcal{T}$. We now wish to expand this ordering by inserting an additional time period $t \notin \mathcal{T'}$. The ordering $O^j_{m}$ is an ordering of $m+1$ elements, created by inserting alternative $t$ in the $j^{th}$ position in the ordering $O_m$. More precisely, all time periods in positions $j$ to $m$ in the ordering $O_m$ are placed one position further back, and time period $t$ is placed in the $j^{th}$ position. The value of the best (partial) rational choice type consistent with $O^j_{m}$ can be evaluated using (\ref{Eq:Order Eval}), if one exists. In this fashion, the best insertion position can be identified and the resulting ordering is fixed. This process is repeated until all time periods have been added to the ordering. \\

In the implementation, we add a dummy patch $x_{t,I_t + 1}$ for each $t \in \mathcal{T}$, with $X_{t,I_t + 1, t'} = 1$ for all $t' \in \mathcal{T}$ and $s_{t,I_t + 1}$ an arbitrarily low (negative) number. In this way, the value $V(O_m)$ is always defined, and negative value indicates that there does not exist a consistent rational choice type.

\begin{algorithm} \caption{Best Insertion Algorithm.}
\label{Alg: Best Insert}
\begin{algorithmic}[1]
\STATE Choose $t \in \mathcal{T}$.
\STATE Create order $O_1$ and set $o_1(t) := 1$.
\STATE Set $\mathcal{T}' := \{t\}$, $k := 1$.
\WHILE{$\mathcal{T}' \ne \mathcal{T}$}
\STATE Choose $t \in \mathcal{T} \backslash \mathcal{T}'$.
\STATE For each $j = 1,\ldots,k$, compute $V(O_k^{j})$.
\STATE Let $r := \arg \max_{j = 1,\ldots,k} V(O_k^{j})$.
\STATE Set $O_{k + 1} := O_{k}^r$.
\STATE Set $\mathcal{T}' := \mathcal{T}' \cup \{t\}$.
\STATE Set $k := k + 1$.
\ENDWHILE
\end{algorithmic}
\end{algorithm}

A final implementation note, is that the time period to be inserted in the partial order can be chosen freely. Different choices in the order in which time periods are inserted can lead to different orderings. In the implementation, we randomly generated the orders in which the periods are inserted. For each pricing iteration we ran the algorithm 10 times with different insertion orders. From these 10 runs of the best insertion algorithm, only the best solution to the pricing problem is kept.
\newpage

\section*{Appendix B: Generation of Choice Types for Tightening}
To tighten the cone based on a subset of the rational choice types, we generate the subset in a semi-random way. First, we generate (likely irrational) choice types by randomly choosing one patch on each budget. If this choice type is rational, we add it to the subset for tightening. If it is not, we identify the subsets of budgets for which preference cycles exist. For each such subset, we randomly pick one budget. For that budget, we look for a patch which, (i) removes at least one preference relations within the subset (ii) is as close as possible to the currently selected patch on that budget (ii) removes, rather than adds revealed preference relations. In this way, we slightly change the choice type, while increasing the chance that it is a rational choice type. If after these changes the choice type is not yet rational, the procedure is repeated until a rational choice type is found. Algorithm \ref{Alg: Gen of Rat Choice Type} contains the pseudo-code to generate these rational choice types in a semi-random way.\\

We define $\mathcal{T}$ as the set of all time periods.

\begin{algorithm} \caption{Generation of rational choice types.}
\label{Alg: Gen of Rat Choice Type}
\begin{algorithmic}[1]
\STATE Randomly generate a choice type $a$ with $\sum_{i = 1}^{I_t} a_{t,i} = 1$.
\WHILE {$a \notin \mathcal{R}$}
\STATE Identify revealed preference relations $r_{i,j},~\forall i,j = 1,\ldots,T$.
\STATE Identify a partitioning $\mathcal{T}_1, \ldots, \mathcal{T}_m$ with $\bigcup_{i=1}^{m} \mathcal{T}_i = \mathcal{T}$ and $\mathcal{T}_i \cap \mathcal{T}_j = \varnothing$ for all $i \ne j$.
\FORALL{$\mathcal{T}_k$ with $|\mathcal{T}_k| > 1$}
\STATE Randomly choose $t \in \mathcal{T}_k$, with $x_{t,z}$ the currently chosen patch on $\mathcal{B}_t$.
\FORALL{$x_{t,i}$, $i = 1, \ldots, I_t$}
\IF{$X_{t,i,t'} \leq X_{t,z,t'}$ for all $t' \in \mathcal{T}_i$}
\STATE $Score_i := 999$.
\ENDIF
\FORALL{$t' \in \mathcal{T}$}
\IF{$X_{t,z,t'} = -1$ and $X_{t,i,t'} = 1$}
\STATE $Score_i := Score_i + 1$.
\ELSIF{$X_{t,z,t'} = 1$ and $X_{t,i,t'} = -1$}
\STATE $Score_i := Score_i + 5$.
\ENDIF
\ENDFOR
\STATE Find a patch $x_{t,j}$ with $j \in \arg \min_{i = 1,\ldots,I_t} Score_i$.
\STATE Set $a_{t,z} := 0$ and $a_{t,j} := 1$.
\ENDFOR
\ENDFOR
\ENDWHILE
\end{algorithmic}
\end{algorithm}
\newpage
\section*{Appendix C: Results Tables}
\begin{table}[htbp]
  \centering
    \begin{tabular}{rrrr|rr|rr|rr|rr|}
          &       &       &       & \multicolumn{2}{c|}{Exact - No Bounds} & \multicolumn{2}{c|}{Heur. - No Bounds} & \multicolumn{2}{c|}{Heur. - UB} & \multicolumn{2}{c|}{Heur. -All Bounds} \\
    \hline
    \multicolumn{2}{c}{Periods} & \multicolumn{1}{l}{Jstat} & \multicolumn{1}{l|}{Pval} & \multicolumn{1}{l}{Time} & \multicolumn{1}{l|}{Completed} & \multicolumn{1}{l}{Time} & \multicolumn{1}{l|}{Completed} & \multicolumn{1}{l}{Time} & \multicolumn{1}{l|}{Completed} & \multicolumn{1}{l}{Time} & \multicolumn{1}{l|}{Completed} \\
    75    & 80    & 0.34  & 0.03  & 7.4   & 1000  & 6.6   & 1000  & 1.9   & 1000  & 2.1   & 1000 \\
    76    & 81    & 0.92  & 0.25  & 8.5   & 1000  & 6.4   & 1000  & 3.1   & 1000  & 3.0   & 1000 \\
    77    & 82    & 0.90  & 0.51  & 7.4   & 1000  & 6.2   & 1000  & 4.0   & 1000  & 3.9   & 1000 \\
    78    & 83    & 0.52  & 0.53  & 9.5   & 1000  & 6.2   & 1000  & 3.9   & 1000  & 3.9   & 1000 \\
    79    & 84    & 0.02  & 0.985 & 11.5  & 1000  & 7.6   & 1000  & 7.7   & 1000  & 7.6   & 1000 \\
    80    & 85    & 0.08  & 0.71  & 7.2   & 1000  & 6.1   & 1000  & 4.9   & 1000  & 4.7   & 1000 \\
    81    & 86    & 0.09  & 0.83  & 9.5   & 1000  & 8.1   & 1000  & 6.8   & 1000  & 6.6   & 1000 \\
    82    & 87    & 0.10  & 0.89  & 12.3  & 1000  & 10.1  & 1000  & 9.3   & 1000  & 9.0   & 1000 \\
    83    & 88    & 0.48  & 0.68  & 7.6   & 1000  & 7.1   & 1000  & 5.5   & 1000  & 5.3   & 1000 \\
    84    & 89    & 0.56  & 0.44  & 4.7   & 1000  & 4.8   & 1000  & 3.0   & 1000  & 2.8   & 1000 \\
    85    & 90    & 0.03  & 0.70  & 2.3   & 1000  & 2.6   & 1000  & 2.2   & 1000  & 2.2   & 1000 \\
    86    & 91    & 1.42  & 0.27  & 3.4   & 1000  & 3.4   & 1000  & 2.0   & 1000  & 2.0   & 1000 \\
    87    & 92    & 2.94  & 0.17  & 4.5   & 1000  & 4.7   & 1000  & 2.2   & 1000  & 2.1   & 1000 \\
    88    & 93    & 1.51  & 0.21  & 2.6   & 1000  & 2.7   & 1000  & 1.8   & 1000  & 1.8   & 1000 \\
    89    & 94    & 1.72  & 0.20  & 1.8   & 1000  & 1.9   & 1000  & 1.5   & 1000  & 1.5   & 1000 \\
    90    & 95    & 0.00  & 1.00  & 1.5   & 1000  & 1.5   & 1000  & 2.5   & 1000  & 2.5   & 1000 \\
    91    & 96    & 0.31  & 0.51  & 3.6   & 1000  & 3.7   & 1000  & 2.7   & 1000  & 2.7   & 1000 \\
    92    & 97    & 0.67  & 0.45  & 5.9   & 1000  & 5.0   & 1000  & 3.0   & 1000  & 2.9   & 1000 \\
    93    & 98    & 0.38  & 0.52  & 9.6   & 1000  & 7.1   & 1000  & 4.8   & 1000  & 4.7   & 1000 \\
    94    & 99    & 0.26  & 0.81  & 13.0  & 1000  & 10.7  & 1000  & 9.6   & 1000  & 9.3   & 1000 \\
    \end{tabular}
    \caption{Computational Results for 6 periods, 3 goods.}
  \label{tab:3-6}%
\end{table}%

\begin{table}[htbp]
  \centering
    \begin{tabular}{rrrr|rr|rr|rr|rr|}
          &       &       &       & \multicolumn{2}{c|}{Exact - No Bounds} & \multicolumn{2}{c|}{Heur. - No Bounds} & \multicolumn{2}{c|}{Heur. - UB} & \multicolumn{2}{c|}{Heur. -All Bounds} \\
    \hline
    \multicolumn{2}{c}{Periods} & \multicolumn{1}{l}{Jstat} & \multicolumn{1}{l|}{Pval} & \multicolumn{1}{l}{Time} & \multicolumn{1}{l|}{Completed} & \multicolumn{1}{l}{Time} & \multicolumn{1}{l|}{Completed} & \multicolumn{1}{l}{Time} & \multicolumn{1}{l|}{Completed} & \multicolumn{1}{l}{Time} & \multicolumn{1}{l|}{Completed} \\
    75    & 84    & 1.79  & 0.59  & 230   & 1000  & 72    & 1000  & 52    & 1000  & 50    & 1000 \\
    76    & 85    & 2.18  & 0.74  & 209   & 1000  & 97    & 1000  & 82    & 1000  & 74    & 1000 \\
    77    & 86    & 1.66  & 0.81  & 207   & 1000  & 87    & 1000  & 79    & 1000  & 76    & 1000 \\
    78    & 87    & 1.79  & 0.66  & 335   & 1000  & 118   & 1000  & 81    & 1000  & 72    & 1000 \\
    79    & 88    & 0.68  & 0.73  & 275   & 1000  & 88    & 1000  & 62    & 1000  & 57    & 1000 \\
    80    & 89    & 4.00  & 0.26  & 76    & 1000  & 17    & 1000  & 15    & 1000  & 14    & 1000 \\
    81    & 90    & 5.38  & 0.37  & 30    & 1000  & 21    & 1000  & 10    & 1000  & 10    & 1000 \\
    82    & 91    & 5.81  & 0.34  & 37    & 1000  & 26    & 1000  & 13    & 1000  & 13    & 1000 \\
    83    & 92    & 3.87  & 0.48  & 78    & 1000  & 38    & 1000  & 20    & 1000  & 20    & 1000 \\
    84    & 93    & 3.88  & 0.38  & 77    & 1000  & 29    & 1000  & 15    & 1000  & 15    & 1000 \\
    85    & 94    & 3.78  & 0.29  & 27    & 1000  & 19    & 1000  & 7     & 1000  & 6     & 1000 \\
    86    & 95    & 3.54  & 0.24  & 19    & 1000  & 14    & 1000  & 5     & 1000  & 5     & 1000 \\
    87    & 96    & 4.99  & 0.25  & 11    & 1000  & 9     & 1000  & 5     & 1000  & 4     & 1000 \\
    88    & 97    & 3.53  & 0.32  & 16    & 1000  & 12    & 1000  & 5     & 1000  & 5     & 1000 \\
    89    & 98    & 3.92  & 0.35  & 12    & 1000  & 10    & 1000  & 5     & 1000  & 5     & 1000 \\
    90    & 99    & 1.12  & 0.65  & 17    & 1000  & 13    & 1000  & 9     & 1000  & 9     & 1000 \\
    \end{tabular}%
  \label{tab:3-10}%
  \caption{Computational Results for 10 periods, 3 goods.}
\end{table}%

\begin{table}[htbp]
  \centering
    \begin{tabular}{rrrr|rr|rr|rr|rr|}
          &       &       &       & \multicolumn{2}{c|}{Exact - No Bounds} & \multicolumn{2}{c|}{Heur. - No Bounds} & \multicolumn{2}{c|}{Heur. - UB} & \multicolumn{2}{c|}{Heur. -All Bounds} \\
    \midrule
    \multicolumn{2}{c}{Periods} & \multicolumn{1}{l}{Jstat} & \multicolumn{1}{l|}{Pval} & \multicolumn{1}{l}{Time} & \multicolumn{1}{l|}{Completed} & \multicolumn{1}{l}{Time} & \multicolumn{1}{l|}{Completed} & \multicolumn{1}{l}{Time} & \multicolumn{1}{l|}{Completed} & \multicolumn{1}{l}{Time} & \multicolumn{1}{l|}{Completed} \\
    75    & 89    & 7.17  & 0.39  &       & 21    & 3557  & 1000  & 1462  & 1000  & 1277  & 1000 \\
    76    & 90    & 8.25  & 0.59  &       & 101   & 2120  & 1000  & 1559  & 1000  & 1327  & 1000 \\
    77    & 91    & 9.90  & 0.49  &       & 214   & 1216  & 1000  & 682   & 1000  & 545   & 1000 \\
    78    & 92    & 11.60 & 0.36  &       & 131   & 1985  & 1000  & 739   & 1000  & 665   & 1000 \\
    79    & 93    & 8.05  & 0.46  &       & 147   & 2890  & 1000  & 1230  & 1000  & 1044  & 1000 \\
    80    & 94    & 8.75  & 0.40  &       & 483   & 1336  & 1000  & 576   & 1000  & 590   & 1000 \\
    81    & 95    & 10.03 & 0.29  & 2906  & 1000  & 785   & 1000  & 229   & 1000  & 243   & 1000 \\
    82    & 96    & 10.09 & 0.29  & 2220  & 1000  & 561   & 1000  & 194   & 1000  & 175   & 1000 \\
    83    & 97    & 6.86  & 0.53  & 1312  & 1000  & 347   & 1000  & 218   & 1000  & 177   & 1000 \\
    84    & 98    & 7.00  & 0.50  & 897   & 1000  & 215   & 1000  & 135   & 1000  & 132   & 1000 \\
    85    & 99    & 6.81  & 0.43  & 249   & 1000  & 81    & 1000  & 44    & 1000  & 42    & 1000 \\
    \end{tabular}%
  \label{tab:3-15}%
  \caption{Computational Results for 15 periods, 3 goods.}
\end{table}%

\begin{table}[htbp]
  \centering
    \begin{tabular}{rrrr|rr|rr|rr|rr|}
          &       &       &       & \multicolumn{2}{c|}{Exact - No Bounds} & \multicolumn{2}{c|}{Heur. - No Bounds} & \multicolumn{2}{c|}{Heur. - UB} & \multicolumn{2}{c|}{Heur. -All Bounds} \\
    \midrule
    \multicolumn{2}{c}{Periods} & \multicolumn{1}{l}{Jstat} & \multicolumn{1}{l|}{Pval} & \multicolumn{1}{l}{Time} & \multicolumn{1}{l|}{Completed} & \multicolumn{1}{l}{Time} & \multicolumn{1}{l|}{Completed} & \multicolumn{1}{l}{Time} & \multicolumn{1}{l|}{Completed} & \multicolumn{1}{l}{Time} & \multicolumn{1}{l|}{Completed} \\
    75    & 94    & 12.57 &       &       & 1     &       & 34    &       & 91    &       & 142 \\
    76    & 95    & 12.73 &       &       & 2     &       & 67    &       & 154   &       & 170 \\
    77    & 96    & 13.84 &       &       & 4     &       & 134   &       & 316   &       & 404 \\
    78    & 97    & 14.76 &       &       & 17    &       & 162   &       & 422   &       & 462 \\
    79    & 98    & 12.06 &       &       & 22    &       & 195   &       & 410   &       & 433 \\
    80    & 99    & 12.93 &       &       & 84    &       & 430   &       & 791   &       & 867 \\
    \end{tabular}%
  \label{tab:3-20}%
  \caption{Computational Results for 20 periods, 3 goods.}
\end{table}%

\newpage
\begin{table}[h!]
  \centering
    \begin{tabular}{rrrrr|rr}
    \multicolumn{1}{l}{\# Goods} & \multicolumn{1}{l}{\# Periods} & \multicolumn{2}{c}{Periods} & \multicolumn{1}{l|}{Pval} & \multicolumn{1}{l}{Computation Time} & \multicolumn{1}{l}{Completed} \\
    \midrule
    4     & 10    & 75    & 84    & 0.627 & 430   & 1000 \\
    4     & 10    & 76    & 85    & 0.772 & 406   & 1000 \\
    4     & 10    & 77    & 86    & 0.676 & 717   & 1000 \\
    4     & 10    & 78    & 87    & 0.985 & 2957  & 1000 \\
    4     & 10    & 79    & 88    &       &       & 529 \\
    4     & 10    & 80    & 89    & 0.361 & 353   & 1000 \\
    4     & 10    & 81    & 90    & 0.546 & 91    & 1000 \\
    4     & 10    & 82    & 91    & 0.706 & 125   & 1000 \\
    4     & 10    & 83    & 92    & 0.989 & 172   & 1000 \\
    4     & 10    & 84    & 93    & 0.98  & 91    & 1000 \\
    4     & 10    & 85    & 94    & 0.946 & 71    & 1000 \\
    4     & 10    & 86    & 95    & 0.799 & 207   & 1000 \\
    4     & 10    & 87    & 96    & 0.783 & 438   & 1000 \\
    4     & 10    & 88    & 97    & 0.809 & 666   & 1000 \\
    4     & 10    & 89    & 98    & 0.817 & 355   & 1000 \\
    4     & 10    & 90    & 99    & 0.794 & 425   & 1000 \\
    5     & 10    & 75    & 84    & 0.353 & 324   & 1000 \\
    5     & 10    & 76    & 85    & 0.687 & 611   & 1000 \\
    5     & 10    & 77    & 86    & 0.531 & 1158  & 1000 \\
    5     & 10    & 78    & 87    &       &       & 479 \\
    5     & 10    & 79    & 88    &       &       & 150 \\
    5     & 10    & 80    & 89    & 0.811 & 1477  & 1000 \\
    5     & 10    & 81    & 90    & 0.89  & 253   & 1000 \\
    5     & 10    & 82    & 91    & 0.93  & 223   & 1000 \\
    5     & 10    & 83    & 92    & 1     & 255   & 1000 \\
    5     & 10    & 84    & 93    & 0.999 & 147   & 1000 \\
    5     & 10    & 85    & 94    & 0.998 & 80    & 1000 \\
    5     & 10    & 86    & 95    & 0.778 & 380   & 1000 \\
    5     & 10    & 87    & 96    & 0.803 & 903   & 1000 \\
    5     & 10    & 88    & 97    & 0.864 & 1053  & 1000 \\
    5     & 10    & 89    & 98    & 0.894 & 648   & 1000 \\
    5     & 10    & 90    & 99    & 0.914 & 636   & 1000 \\
    4     & 15    & 75    & 89    &       &       & 10 \\
    4     & 15    & 76    & 90    &       &       & 40 \\
    4     & 15    & 77    & 91    &       &       & 59 \\
    4     & 15    & 78    & 92    &       &       & 35 \\
    4     & 15    & 79    & 93    &       &       & 17 \\
    4     & 15    & 80    & 94    &       &       & 27 \\
    4     & 15    & 81    & 95    &       &       & 29 \\
    4     & 15    & 82    & 96    &       &       & 23 \\
    4     & 15    & 83    & 97    &       &       & 31 \\
    4     & 15    & 84    & 98    &       &       & 39 \\
    4     & 15    & 85    & 99    &       &       & 44 \\
    5     & 15    & 75    & 89    &       &       & 9 \\
    5     & 15    & 76    & 90    &       &       & 13 \\
    5     & 15    & 77    & 91    &       &       & 13 \\
    5     & 15    & 78    & 92    &       &       & 9 \\
    5     & 15    & 79    & 93    &       &       & 6 \\
    5     & 15    & 80    & 94    &       &       & 12 \\
    5     & 15    & 81    & 95    &       &       & 13 \\
    5     & 15    & 82    & 96    &       &       & 9 \\
    5     & 15    & 83    & 97    &       &       & 17 \\
    5     & 15    & 84    & 98    &       &       & 22 \\
    5     & 15    & 85    & 99    &       &       & 17 \\
    \end{tabular}%
  \label{tab:4&5}%
  \caption{Computation times for 4 and 5 goods.}
\end{table}%


\end{document}